\documentclass[a4paper,10pt]{llncs}

\usepackage[utf8]{inputenc}
\usepackage[english]{babel}
\usepackage[cmex10]{amsmath}
\usepackage{amssymb}
\usepackage{color}
\usepackage{graphicx}

\title{Combinatorial metrics: MacWilliams-type identities, isometries and extension property} 
\author{Jerry Anderson Pinheiro \and Roberto Assis Machado \and Marcelo Firer} 
\institute{University of Campinas - Institute of Mathematics, Statistics ans Scientific Computing\\ 
\email{jerryapinheiro@gmail.com, robertomachado@ime.unicamp.br, mfirer@ime.unicamp.br}}

\begin{document} 
	\maketitle{} 
\begin{abstract}
In this work we characterize the combinatorial metrics admitting a MacWilliams-type identity and describe the group of linear isometries of such metrics. Considering coverings that are not  connected, we classify the metrics satisfying the MacWilliams extension property.
\end{abstract}

\section{Introduction}

In the context of coding theory, different metrics have been used  to provide good (efficient) alternatives to the Maximum a Posteriori Decoders (MAP), which is the ideal observer decoder determined by the distribution probabilities of a given channel. In general, the quality of an encoder is measured by its usefulness and its manageability. 
Due to its structure, many metric decoders (Minimum Distance Decoders - MD) simplify the decoding process. 
The family of combinatorial metrics  attend some usefulness condition: “The b-burst metric can be considered as a combinatorial metric” \cite{gabidulin1973combinatorial}. The study of combinatorial metrics rested nearly untouched since its introduction in 1973 and just recently, after they were recalled in a survey made by Gabidulin in 2012, the interest in these metrics arose. In order to determine the manageability of such metrics, it is necessary to explore the details of the geometry. This is the direction we work here.


Some subfamilies of combinatorial metrics has been widely explored in the literature, as we can see, for example, the block and translational metrics in \cite{feng2006linear} and \cite{mohamedcombinatorial}, respectively. In a general setting, a few number of papers are devoted to these metrics, as one of the exceptions we can cite the work \cite{bossert1996singleton} concerning Singleton-type bounds. Classical coding properties like MacWilliams’ Identities and MacWilliams’ Extensions have not been yet explored in the general case. Our objective is to characterize the combinatorial metrics having a MacWillimas-type Identity and to describe the group of linear isometries of such metrics. Although we still have not classified the metrics with the extension property, we managed to get a partial characterization for the combinatorial metrics determined by not connected coverings. We expect that these partial results and the description of the group of linear isometries may lead us to the complete characterization of combinatorial metrics having the MacWilliams extension property.

This work is organized as follows. In Section \ref{prel} we define the combinatorial metric and the redundancy of a covering. In Section \ref{iden} we characterize the combinatorial metrics that admits a MacWilliams' identity. In Section \ref{isom} we characterize  the group of linear isometries of a space endowed with a combinatorial metric. In Section \ref{ext} we give necessary and sufficient conditions for an unconnected covering to determine a metric which satisfies an extension property of isometries, similar to the MacWilliams    Extension Theorem. 
We remark that, due to lack of space, some technical proofs are omitted and some are shortened.


\section{Preliminaries}\label{prel}

Let $\mathbb{F}_q^n$ be the $n$-dimensional vector space over the  field $\mathbb{F}_q$, $[n]:=\{1,\ldots,n\}$ and $\mathcal{P}_{n}:= \{A: A\subset [n] \}  $  the power set of $[n]$. We say that a family $\mathcal{A} \subset \mathcal{P}_{n}$ is a \textit{covering}  of a set $X\subset [n]$ if, and only if, $X\subset \cup_{A\in\mathcal{A} }A$. 
If $\mathcal{F}$ is a covering of $[n]$, then the $\mathcal{F}$\textit{-combinatorial weight} of $x=(x_1,\ldots,x_n)\in\mathbb{F}_q^n$ is the integer-valued map $\mathrm{wt}_\mathcal{F}$ defined by
\[
	\mathrm{wt}_{\mathcal{F}}(x)=\min \{|\mathcal{A}| \ : \ \mathcal{A}\subset \mathcal{F}  \text{ and }  \mathcal{A} \text{ is a covering of } \mathrm{supp}(x)\}, 
\] 
where $ \mathrm{supp}(x)=\{i\in[n]: x_i\neq 0  \}$ is the \emph{support} of $x$. Each element $A\in \mathcal{F}$ is called a \emph{basic set} of the covering.

As showed in \cite{gabidulin1973combinatorial}, the function $d_{\mathcal{F}}:\mathbb{F}_q^n\times \mathbb{F}_q^n \rightarrow \mathbb{N}$ defined by
\[
	d_\mathcal{F}(x,y)=\mathrm{wt}_{\mathcal{F}}(x-y)
\]
satisfies the metric axioms and is called $\mathcal{F}$\textit{-combinatorial metric}. 


\begin{example}[Block Metrics, \cite{feng2006linear}]
Suppose $\mathcal{F}$ is a partition of $[n]$, that is, the  basic sets are pairwise disjoint.  In this case, the $\mathcal{F}$-combinatorial metric is also called a \emph{block metric}.  In the particular case that every basic set has a unique element ($\mathcal{F}=\{\{1\},\{2\},\ldots,\{n\}\}$), we have the classical Hamming metric.
\end{example}

\begin{example}[$b$-burst Metric, \cite{6774193}]
Given an integer $b$, denote $[b]+i=\{1+i,2+i,\ldots,b+i\}$. Let 
$$\mathcal{F}=\{[b],[b]+1,[b]+2,\ldots,[b]+(n-b)\}$$ 
be the partition over $[n]$ where $b<n$. The metric induced by $\mathcal{F}$ is called the $b$\emph{-burst metric}. 
\end{example}

Note that both the coverings $\mathcal{F}_1=\{[n]\}$ and $\mathcal{F}_2=\{[n],B\}$ where $B\subset[n]$ is any subset, determine the same metric, indeed, for every $x,y\in\mathbb{F}_q^n$, 
$$
d_{\mathcal{F}_1}(x,y)=d_{\mathcal{F}_2}(x,y)=\left\{
\begin{array}
[c]{c}%
0\text{ if } x=y\\
1\text{ if } x\neq y
\end{array}
\right..
$$
In order to eliminate multiplicity (different coverings determining the same metric), we need to define the redundancy of basic sets: given a covering $\mathcal{F}$, we say that $A\in\mathcal{F}$ is $\mathcal{F}$-\emph{redundant} (or just \emph{redundant}) if there is $B\in\mathcal{F}$, with $A\subset B$ and $A\neq B$. We denote by $\overline{\mathcal{F}}$ the set of all redundant basic sets. 

\begin{proposition}\label{prop01}
	Given a covering $\mathcal{F}$ of $[n]$, the set $\mathcal{F}_2=\mathcal{F}\setminus\overline{\mathcal{F}}$ is also a covering of $[n]$ and determines the same combinatorial metric of $\mathcal{F}$.
\end{proposition}

\begin{proof}
	Follows straightforward from the definitions.
\end{proof}

From Proposition \ref{prop01}, we may (and will) assume that $\mathcal{F}$ has no redundancy.

\begin{proposition}
	Two different coverings with no redundancy determine different metrics. 
\end{proposition}

\begin{proof}
		Follows straightforward from the definitions.
\end{proof}

We end this section with a definition which will used many times later.

\begin{definition}\label{Fk}
	A covering $ \mathcal{F}$ is called a \emph{$k$-partition} if it is a partition of $[n]$  and every $A\in \mathcal{F}$ has constant cardinality $k=|A|$. In this case, the $\mathcal{F}$-combinatorial metric is  called an $(\mathcal{F},k)$-combinatorial metric.
\end{definition}



\section{MacWilliams’ Identities}\label{iden}

The classical MacWilliams identity, presented in \cite{macwilliams1963theorem}, is a remarkable result in coding theory that relates, in the case of the Hamming metric, weight enumerators of codes and weight enumerators of their duals. When another metric is in place, to establish such relations may not be possible, as we can see in the counterexamples for the Lee metric constructed in \cite{shi2015note} and in the classification of poset-block metrics admitting a MacWilliams-type identity presented in \cite{pinheiro2012classification}.

Regarding combinatorial metrics, the block metrics is the unique instance where the MacWilliams identities were completely described. For the general case, it is not known if it is possible to obtain such identities. 

The dual of a linear code $\mathcal{C}\subset \mathbb{F}_q^n$ is the space $\mathcal{C}^\perp = \{u\in\mathbb{F}_q^n \ : \ u\cdot c=\sum_{i=1}^n u_ic_i =0 \text{ for every } c\in\mathcal{C}\}$.

The $\mathcal{F}$- \emph{weight enumerator} of a code $\mathcal{C}$ is the polynomial
	\[
		W^\mathcal{F}_\mathcal{C}(x,y)=\sum_{c\in\mathcal{C}} x^{D-wt_{\mathcal{F}}(c)}y^{wt_{\mathcal{F}}(c)}=\sum_{i=0}^D A_i x^{D-i}y^i
	\]
where $A_i=|\{c\in \mathcal{C} \ : \ w(c)=i\}|$ and $D=\max \{ \mathrm{wt}_{\mathcal{F}}(c) \ : \ c\in\mathcal{C} \}$. 
 When no confusion may arise, we write  $	W_\mathcal{C}(x,y)$, omitting the index $\mathcal{F}$.

\begin{definition}
	A combinatorial metric $ d_{\mathcal{F}}$ admits a MacWilliams-type identity if the $\mathcal{F}$-weight enumerator  of a code determines the $\mathcal{F}$-weight enumerator of its dual, i.e., if $	W_{\mathcal{C}_1}(x,y) =	W_{\mathcal{C}_2}(x,y)$  then $	W_{\mathcal{C}^\perp_1}(x,y) =	W_{\mathcal{C}^\perp_2}(x,y)$.
\end{definition}

Restating the results of \cite{feng2006linear} in terms of combinatorial metrics, we have the following:

\begin{proposition}\cite{feng2006linear}\label{prop001}
	Suppose $\mathcal{F}$ is a partition of $[n]$. The combinatorial metric $ d_{\mathcal{F}}$ admits a MacWilliams’ identity if, and only if, $ d_{\mathcal{F}}$ is an $(\mathcal{F},k)$-combinatorial metric for some $k\in\mathbb{N}$.
\end{proposition}

Our goal is to proof that these are all the combinatorial metrics satisfying a MacWilliams identity. 



	\begin{proposition}\label{prop002}
		Let $ d_{\mathcal{F}}$ be a combinatorial metric. If $ d_{\mathcal{F}}$ satisfy a MacWilliams-type identity then $ \mathcal{F}$ is an $(\mathcal{F},k)$-combinatorial metric for some $k$.

	\end{proposition}
	\begin{proof}
		If $\mathcal{F}$ is a partition of $[n]$, the result follows from Proposition \ref{prop001}.
		Suppose $\mathcal{F}$ is not a partition, i.e., $ \mathcal{F}$ is not an $(\mathcal{F},k)$-combinatorial metric, hence, there are $A,B\in \mathcal{F}$ such that $A\cap B \neq \emptyset$ and $A\neq B$ and let $i_1\in A\cap B$. We shall prove that  $d_{\mathcal{F}}$ does not satisfy a MacWilliams-type identity. 
		
		Assuming that $\mathcal{F}$ has no redundancy we find that there is $i_0\in A\setminus B$. Consider the unidimensional codes over $\mathbb{F}_q^n$ given by $\mathcal{C}_1=span\{e_{i_0}\}$ and $\mathcal{C}_2=span\{e_{i_0}+e_{i_1}\}$.
\begin{figure}[h]
	\centering
	\includegraphics[scale=0.6]{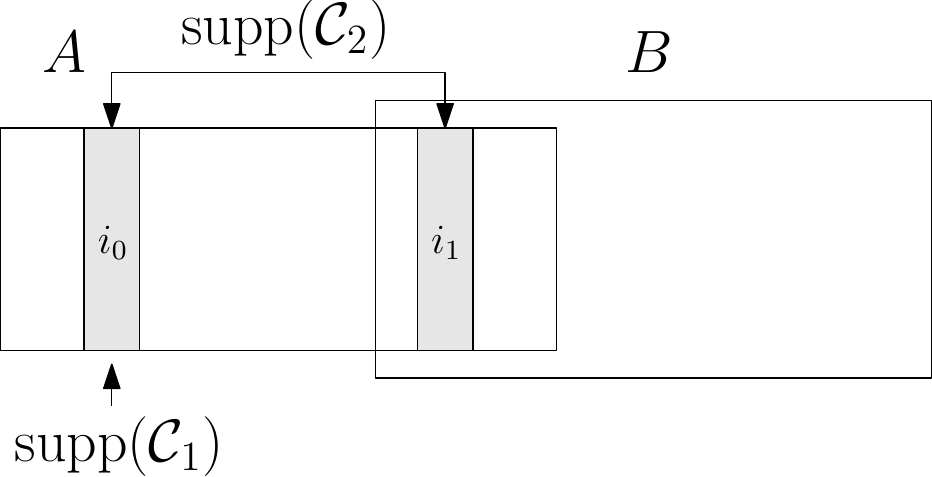}
\end{figure}
		
		\noindent By direct computations we conclude that $$	W_{\mathcal{C}_1}(x,y) =	W_{\mathcal{C}_2}(x,y)=1+(q-1)X^{D-1}y.$$


		Given $c=(c_1,\cdots ,c_n)\in\mathcal{C}_1^\perp$, since $c_{i_0}=0$, we get that $c-c_{i_1}e_{i_0}\in\mathcal{C}_2^\perp$ and hence we get a linear map 
		$T:\mathcal{C}_1^\perp  \rightarrow \mathcal{C}_2^\perp$ by setting $T(c)=c-c_{i_1}e_{i_0}$. By construction, $T$ is an injection, hence a bijection.

It is possible to prove that  $T$ preserves weight ($\mathrm{wt}_{\mathcal{F}}(x)=\mathrm{wt}_{\mathcal{F}}(T(x))$ for every $x\in\mathcal{C}_1^\perp$) if, and only if, $	W_{\mathcal{C}^\perp_1}(x,y) =	W_{\mathcal{C}^\perp_2}(x,y)$.

If $c=\sum_{i\in B} e_i$, then $\mathrm{supp}(c)=B$ and $\mathrm{wt}_{\mathcal{F}}(c)=1$, furthermore, $c\in\mathcal{C}_1^\perp$ because $i_0\not\in B$. Since $\{i_0\}\cup B = \mathrm{supp}(T(c))$ and $\mathcal{F}$ has no redundancy, it follows that $\mathrm{wt}_{\mathcal{F}}(T(c))>1$. Thus, $T$ does not preserve weight and it follows that $W_{\mathcal{C}_1^\perp}(x,y)\neq W_{\mathcal{C}_2^\perp}(x,y)$,  
hence $d_{\mathcal{F}}$ does not satisfy a MacWilliams-type identity.

	\end{proof}

The following theorem is a direct consequence of Propositions \ref{prop001} and \ref{prop002}.

\begin{theorem}
	A combinatorial metric $ d_{\mathcal{F}}$ admits a MacWilliams-type Identity if, and only if, $ d_{\mathcal{F}}$ is an $(\mathcal{F},k)$-combinatorial metric.
\end{theorem}

\section{Linear $\mathcal{F}$-isometries}\label{isom}

In the context of coding theory, the linear group of isometries has been characterized considering many different metrics (see for example \cite{panek,tuvi}) and been used as a relevant tool to prove coding related results (see \cite{luvifelix,robertoassis}). We aim to characterize the group of linear isometries of a space endowed with a combinatorial metric. We start with some definitions. 

Let us denote by $GL\left(  n,\mathcal{F}\right)  _{q}$ the group of linear
isometries of $\left(  \mathbb{F}_{q}^{n},d_{\mathcal{F}}\right)  $, i.e.,
\begin{align*}
GL\left(  n,\mathcal{F}\right)_{q}=\{   T: \mathbb{F}_{q}^{n}\rightarrow
\mathbb{F}_{q}^{n}\ :& \ T \text{ is linear and } \\
& d_{\mathcal{F}}\left(  x,y\right)  =d_{\mathcal{F}}\left(T  \left( x\right) ,T\left( y\right)  \right), \forall x,y\in \mathbb{F}_{q}^{n} \}.
\end{align*}

\begin{definition}
	Let $\mathcal{F}$ be a covering of $[n]$. We say that a permutation $\phi:\left[  n\right]  \rightarrow\left[  n\right]$ \emph{preserves} $\mathcal{F}$  if $\phi(A)\in\mathcal{F}$, for every $A\in\mathcal{F}$.
\end{definition}

Let $S_n$ be the group of permutations of  $[n]$. Consider the action of $S_n$ on $\mathbb{F}_{q}^{n}$ by permutation of coordinates: given $\phi\in S_n$, we define a map $T:\mathbb{F}_q^n \rightarrow \mathbb{F}_q^n$  as 
$
T_{\phi
}\left(  (  x_{1},x_{2},\ldots,x_{n})  \right)  =(
x_{\phi\left(  1\right)  },x_{\phi\left(  2\right)  },\ldots,x_{\phi\left(
	n\right)  }).
$
%

The first thing we remark is that if $\phi$ preserves $\mathcal{F}$ then $T_\phi \in GL\left(  n,\mathcal{F}\right)_{q}$.

\begin{proposition}
	If $\phi$ preserves $\mathcal{F}$ then $T_\phi$ is a 	is a linear $\mathcal{F}$-isometry. 
\end{proposition}
\begin{proof} 	Follows straightforward from the definitions.
	
\end{proof}

We denote $G:=\{ T_\phi: \phi \text{ preserves } \mathcal{F} \}$. 
A covering $\mathcal{F}$ determines a equivalence relation $\sim_{\mathcal{F}}$ on $[n]$ by the following rule:
\begin{center}
	$i \sim_{\mathcal{F}} j$, if $i\in A_k \iff j \in A_k$.
\end{center}
We denote by $H_1,\cdots ,H_s$ the equivalence classes, so we write $[n]= \bigsqcup_{i=1}^s H_i$.

We stress that if an element of a equivalence class $H_i$ belongs to a basic set $A_j\in\mathcal{F}$, then the entire class $H_i$ is contained in $A_j$. Assuming that, $\mathcal{F} = \{A_1, \ldots, A_r\}$ and $H=\{H_1, \ldots, H_s\}$, let $M=M(\mathcal{F}; H)$ be an incidence matrix, defined as follows
\[
m_{ij}=\left\{
\begin{array}
[c]{c}%
1\text{, if }H_{i}\subset A_{j}\\
0\text{, otherwise.}%
\end{array}
\right.  .
\]
We say that a class $H_i$ \textit{dominates} $H_j$ if $\mathrm{supp}(v^j) \subset \mathrm{supp}(v^i)$, where $v^k$ denotes the $k$-th row of $M$. This is an order relation and we denote it by $H_j\leq H_i$. We say that $H_i$ is a \emph{head} in a family of equivalence classes if it is a maximal element in the family.

Given a subset $X\subset [n]$, there is a minimum set $\mathcal{H} = \{H_{i_1}, \ldots, H_{i_k}\}$ of equivalence classes of $\mathcal{F}$ such that $X\subset H_{i_1}\cup \cdots\cup H_{i_k}$. From $\mathcal{H}$, we construct the subset $\widetilde{\mathcal{H}}$ consisting of all the heads in $\mathcal{H}$. 
The  \emph{Minimum Set Header} (MSH) of $X$ is $\widetilde{X}=\{i\in X: i\in H_j \text{ for some } H_j\in\widetilde{\mathcal{H}}\}$.



Given $x=(x_1, \ldots, x_n)\in\mathbb{F}_q^n$ the \emph{cleared out} form of $x$ is the vector $\widetilde{x} = (\widetilde{x}_1, \ldots, \widetilde{x}_n)$ where $\widetilde{x}_i = x_i$ if $i\in \widetilde{\mathrm{supp}(x)}$ and $\widetilde{x}_i=0$ otherwise.

Let $h_i = |H_i|$ be the cardinality of $H_i$, $N_i = \sum_{j=1}^{i}h_j$ and $N_0 = 0$. Without loss of generality, we may relabel the elements of the equivalence classes by $H_i = [N_{i-1}+1, N_i] = \{N_{i-1}+1, \ldots, N_i\}$. 

An $n\times  n$-matrix $B=(b_{xy})$ with coefficients in $\mathbb{F}_q$ is said to respect $M$ if for every block $B_{ij}=(b_{xy})_{x\in H_i, y\in H_j}$, the following conditions hold:

\begin{enumerate}
	\item Each block $B_{ii}=(b_{xy})_{x,y\in H_i}$ is an invertible matrix; 
	\item If $x\in H_i$ and $y \in H_j $ for $i \neq j$, then $B_{ij} \neq 0$ implies $H_j$ dominates $H_i$.
\end{enumerate}

We denote by $K_M$ as the set of all matrices respecting $M$.

%

\begin{proposition}\label{km_prop}
	An $n\times n$ matrix $B$ respecting $M$ is a linear $\mathcal{F}$-isometry, i.e.,  $B\in GL\left(  n,\mathcal{F}\right)_{q}$. 
\end{proposition}
\begin{proof}
	If a vector $x\in\mathbb{F}_q^n$ has $\mathrm{wt}_{\mathcal{F}}(x)=k$ then, there are $A_1, \ldots, A_k\in \mathcal{F}$ covering the support of $x$, i.e., $\mathrm{supp}(x) \subset A_1 \cup \cdots \cup A_k$. Since $B$ respects $M$, every covering of $\mathrm{supp}(x)$ also covers $\mathrm{supp}(Bx)$, that is, $\mathrm{wt}_{\mathcal{F}}(Bv)\leq \mathrm{wt}_{\mathcal{F}}(v)$. It is possible to prove that  $B\in K_M$ implies that  $B^{-1}\in K_M$ and so we have that $\mathrm{wt}_{\mathcal{F}}(v) = \mathrm{wt}_{\mathcal{F}}(Bv)$.
\end{proof}

The previous proposition ensures that for every vector $x\in\mathbb{F}_q^n$ there is a linear $\mathcal{F}$-isometry $S \in K_M$ such that  $S(x)$ the cleared out of $x$, i.e., $S(x)=\widetilde{x}$.


\begin{lemma}\label{anterior}
	Let $T\in GL(n, \mathcal{F})_q$. Given $e_i\in \mathbb{F}_q^n$, the support of $\widetilde{T(e_i)}$ is contained in some equivalence class of $\mathcal{F}$. 
\end{lemma}

\begin{proof} Due to lack of space, this proof is omitted.
\end{proof}

The previous Lemma ensures the existence of an equivalence class that contains $\mathrm{supp}(\widetilde{T(e_i)})$. In the next lemma we prove that this class does not depend on $i$, but only on $T$ and the class containing $i$.


\begin{lemma}\label{anterior2} Given an equivalence class $H$ and an $\mathcal{F}$-isometry $T$, there is an equivalence class $H'$ such that, for every $i\in H$,    $\mathrm{supp}(\widetilde{T(e_i)})\subset H'$.
\end{lemma}
\begin{proof}
	Suppose that $\mathrm{supp}(\widetilde{T(e_i)}) \subset H_1$ and $\mathrm{supp}(\widetilde{T(e_j)}) \subset H_2$ with $H_1 \neq H_2$. It is possible to prove that $\widetilde{T(e_i+e_j)}$ is contained in a unique equivalence class and this implies that either $H_1$ dominates $H_2$ or $H_2$ dominates $H_1$. Let us assume that $H_1$ dominates $H_2$. It means  there is a vector $e_k\in\mathbb{F}_q^n$ such that the $\mathcal{F}$-weight of $T(e_i) + T(e_k)$ is 2 while $T(e_j) + T(e_k)$ has  $\mathcal{F}$-weight 1. It is a contradiction because, by construction, the  vectors $e_i + e_k$ and $e_j + e_k$  have the same $\mathcal{F}$-weight.
\end{proof}


\begin{theorem}
	$GL(n, \mathcal{F}) = G K_M$.
\end{theorem}
\begin{proof}
	Lemmas \ref{anterior} and \ref{anterior2} ensure that for each equivalence class $H_i$ of $\mathcal{F}$, there is $S_i\in K_M$ such that $S_i(T(e_j) = \widetilde{T(e_j)})$ for $j\in H_i$ and $S_i(\widetilde{T(e_k)}) = \widetilde{T(e_k)}$ for every $k\notin H_i$. It follows that, given $T\in GL(n, \mathcal{F})_q$, $S_1S_2\cdots S_nT$ is  a permutation of basic sets of $\mathcal{F}$.
\end{proof}

\section{MacWilliams’ Extension Property}\label{ext}

When working with equivalence relations among linear codes, there are two distinct approaches, a local one and a global one. For the Hamming metric, F. J. MacWilliams, in her thesis (see \cite{macwilliams1962combinatorial}), proved that the two approaches are equivalent. To be more precise, we need some definitions. 

\begin{definition}
	(Local Equivalence) Given an $\mathcal{F}$-combinatorial weight over $\mathbb{F}_q^n$. Two linear codes $\mathcal{C}_1, \mathcal{C}_2 \subset \mathbb{F}_q^n$ are said \emph{locally $\mathcal{F}$-equivalent} if there exist a weight-preserving linear map (local $\mathcal{F}$-equivalence) $t:\mathcal{C}_1\rightarrow \mathcal{C}_2$.
\end{definition}

\begin{definition}
	(Global Equivalence) Two linear codes $\mathcal{C}_1$ and $\mathcal{C}_2$ are said globally $\mathcal{F}$-\emph{equivalent}, or just $\mathcal{F}$-\emph{equivalent}, if there exist a linear isometry ($\mathcal{F}$-equivalence) $T:\mathbb{F}_q^n\rightarrow \mathbb{F}_q^n$ such that $T(\mathcal{C}_1)=\mathcal{C}_2$.
\end{definition}

The MacWilliams result states that, in the Hamming metric case, every weight-preserving linear map $t:\mathcal{C}_1\rightarrow\mathcal{C}_2$ can be extended to a monomial map, hence, in particular, if $\mathcal{F}$ induces the Hamming metric, then two codes are locally $\mathcal{F}$-equivalent if, and only if, they are $\mathcal{F}$-equivalent.
 	
\begin{definition}
	(MacWilliams’ Extension Property - MEP) An $\mathcal{F}$-combinatorial metric satisfies the MacWilliams Extension Property if for any linear codes $\mathcal{C}_1$ and $\mathcal{C}_2$, every local $\mathcal{F}$-equivalence $t:\mathcal{C}_1\rightarrow \mathcal{C}_2$ can be extended to an $\mathcal{F}$-equivalence $T:\mathbb{F}_q^n\rightarrow\mathbb{F}_q^n$, i.e., $T(c)=t(c)$ for every $c\in\mathcal{C}_1$.
\end{definition}




\begin{proposition}\label{samelength}
	If there are $A,B\in\mathcal{F}$ such that $|A|\neq |B|$, then the $\mathcal{F}$-combinatorial metric does not satisfies MEP.
\end{proposition}

\begin{proof}
Let $A,B\in\mathcal{F}$ such that $|A|>|B|$. Define $C\subset [n]$ such that $C\subset A$, $A\cap B \subset C$ and $|C|=|B|$. Let $\sigma:B\setminus A\rightarrow (A\setminus B) \cap C$ be a bijection. Define the linear map $t$ by $t(e_i)=e_i$ for every $i\in A\cap B$ and $t(e_i)=e_{\sigma(i)}$ for every $i\in B\setminus A$. By construction, $t$ is a local $\mathcal{F}$-equivalence. Given $i_0\in A\setminus C$, if $T$ is a linear extension of $t$, then 
\[
	  	T\left(\sum_{j\in C}e_{j}+e_{i_0}\right)=\sum_{j\in B}e_j+T(e_{i_0}).
	  \] 
Since $\mathcal{F}$ has no redundancy, $\mathrm{wt}_{\mathcal{F}}(\sum_{j\in B}e_j+T(e_{i_0}))>1$. Therefore, $T$ is not an isometry.
\end{proof}

In order to characterize the unconnected coverings satisfying the MacWilliams extension property, we need the definition of connected components.

\begin{definition}
  	A covering $\mathcal{F}$ is said to be \emph{connected} if there is no $\mathcal{A},\mathcal{B}\subset \mathcal{F}$ such that $\mathcal{A}\cup \mathcal{B}=\mathcal{F}$ where $A\cap B= \emptyset$ for every $A\in\mathcal{A}$ and $B\in\mathcal{B}$. A connected subset $\mathcal{A}\subset\mathcal{F}$ is called \emph{connected component} of $\mathcal{F}$.
  \end{definition}

\begin{proposition}
	If $\mathcal{F}$ has more than $2$ connected components, then the $\mathcal{F}$-combinatorial metric satisfies MEP if, and only if, it is the Hamming metric. 
\end{proposition}

\begin{proof}
	It is well known that the Hamming metric satisfies MEP. For the opposite direction, we may suppose, without loss of generality, that $\mathcal{F}$ has exactly $3$ connected components: $\mathcal{A}_1$, $\mathcal{A}_2$ and $\mathcal{A}_3$. Furthermore, suppose $d_\mathcal{F}$ does not coincide with the Hamming metric, hence by Proposition \ref{samelength}, $|A|>1$ for every $A\in\mathcal{F}$. Take $A\in\mathcal{A}_1$, $B \in \mathcal{A}_2$ and $C\in\mathcal{A}_3$. Furthermore, take $a_0,a_1\in A$, $b_0,b_1\in B$ and $c\in C$ where $a_0\neq a_1$ and $b_0\neq b_1$. Define     	
	\[
		t(e_{a_0}+e_{b_0})=e_{a_1}+e_{c} \ \
\text{ and } \ \ 
		t(e_{a_1}+e_{b_1})=-e_{a_1}+e_{b_1}.
	\]
By construction, $t$ is a local $\mathcal{F}$-equivalence. Let $T$ be a linear extension of $t$, since  
\begin{equation}\label{eco1}
	T(e_{a_0})=e_{a_1}+e_{c}-T(e_{b_0})
\end{equation}
and
\begin{equation}\label{eco2}
	T(e_{a_1}+e_{b_1}+e_{b_0})=-e_{a_1}+e_{b_1}+ T(e_{b_0}),
\end{equation}
Equation (\ref{eco1}) ensures that $T$ is an isometry if either $T(e_{b_0})=e_{a_1}$ or $T(e_{b_0})=e_{c}$, but in both the cases we get a contradiction by (\ref{eco2}). Therefore, $T$ can not be an $\mathcal{F}$-equivalence. 
\end{proof}

\begin{proposition}
	Suppose $\mathcal{F}$ has two connected components. A combinatorial metric $d_\mathcal{F}$ satisfies MEP if, and only if, $\mathcal{F}$ is a $k$-partition.  
\end{proposition}

\begin{proof}
Suppose that $d_\mathcal{F}$ satisfies MEP and that $\mathcal{F}$ has two connected components $\mathcal{A}_1$ and $\mathcal{A}_2$. By Proposition \ref{samelength} we have that $|A|=|B|$ for every $A,B\in\mathcal{F}$. 
Suppose that $|\mathcal{A}_1|>1$. Thus, there exist $A,B\in\mathcal{A}_1$ such that $A\cap B\neq \emptyset$. Take $C\in\mathcal{A}_2$ and define
\[
	u=\sum_{i\in A} e_i, \ \ \ v=\sum_{i\in B\setminus A} e_i \ \ \ \text{and} \ \ \  w=\sum_{i\in A\cap B} e_i.
\]
Define $t$ by $t(u)=e_{j_0}$ and $t(v)=e_{i_0}$ where $i_0\in A\setminus B$ and $j_0\in C$. By construction $t$ is a local $\mathcal{F}$-equivalence. Suppose $T$ is an $\mathcal{F}$-equivalence and an extension of $t$. Note that if $M=\sum_{i\in B} e_i$, then 
\[
 	T(M)=T(v)+T(w)=e_{i_0}+T(w).
 \] 
 Since $\mathrm{wt}_{\mathcal{F}}(T(M))=1$, it follows that $\mathrm{supp}(T(w))\subset \cup_{A\in\mathcal{A}_{i}} A$ if, and only if, $\mathrm{supp}(e_{i_0})\subset \cup_{A\in\mathcal{A}_{i}} A$. On the other hand, if $N=\sum_{i\in A\setminus B} e_i$, then
 \[
   	T(N)=T(u)-T(w)=e_{j_0}-T(w).
   \]  
   Since $\mathrm{wt}_{\mathcal{F}}(T(N))=1$, it follows that $\mathrm{supp}(T(w))\subset \cup_{A\in\mathcal{A}_{j}} A$ if, and only if, $\mathrm{supp}(e_{j_0})\subset \cup_{A\in\mathcal{A}_{j}} A$. But $\mathrm{supp}(e_{j_0})\subset \cup_{A\in\mathcal{A}_{j}} A$ and $\mathrm{supp}(e_{i_0})\subset \cup_{A\in\mathcal{A}_{j}} A$ with $i\neq j$. Hence, $T$ can not be an isometry. The other implication is a lengthy and delicate construction of the desired extension. Due to the limitations of space, it will be omitted.
\end{proof}

 Summarizing the previous results, we have the following theorem:
\begin{theorem}
	If $\mathcal{F}$ is unconnected with $l$ connected components, $d_\mathcal{F}$ satisfies MEP if, and only if, either $l=2$ and $\mathcal{F}$ is a $k$-partition or $l>2$ and $d_\mathcal{F}$ is the Hamming metric.
\end{theorem}

To complete the characterization of the combinatorial metrics satisfying MEP, the case of combinatorial metrics determined by connected coverings must be solved. Based on some examples and on the characterization of the group of linear isometries, we have the following conjecture:

\begin{conjecture}
	Suppose $\mathcal{F}$ is connected. The metric $d_\mathcal{F}$ satisfies MEP if, and only if, $|A|=|B|$ for every $A,B\in\mathcal{F}$ and $C\in\mathcal{F}$ for any $C\subset [n]$ with $|C|=|A|$. 
\end{conjecture}

 \section*{Acknowledgment}
The authors would like to thank the S\~{a}o Paulo Research Foundation (FAPESP) for the financial support through grants 2013/25977-7, 2015/11286-8 and 2016/01551-9.

    \bibliographystyle{plain} 
    \bibliography{biblio}

\end{document}